\documentclass[12pt, leqno]{article}

\usepackage{amsthm} \usepackage{amsmath} \usepackage{amsfonts}
\usepackage{amssymb} \usepackage{latexsym} \usepackage{enumerate}
\usepackage[dvips]{color}
 \usepackage{mathrsfs}
\usepackage[numbers]{natbib}

\theoremstyle{plain}%
\newtheorem{Theorem}{Theorem}[section] %
\newtheorem{Lemma}[Theorem]{Lemma}
\newtheorem{Proposition}[Theorem]{Proposition} %

\theoremstyle{definition}%
\newtheorem{Assumption}[Theorem]{Assumption}%
\newtheorem{Definition}[Theorem]{Definition}%
\newtheorem{Example}[Theorem]{Example} %

\theoremstyle{remark}%
\newtheorem{Remark}[Theorem]{Remark} %

\newcommand{\set}{\triangleq}

\DeclareMathOperator*{\argmin}{arg\:min}

\newcommand{\envspace}{\vspace{2mm}}

\renewcommand{\mathcal}{\mathscr}

\renewcommand{\epsilon}{\varepsilon}

\newcommand{\esssup}{\operatorname*{\mathrm{ess\,sup}}}

\numberwithin{equation}{section}
\begin{document}

\title{Necessary and sufficient conditions in the problem of optimal
  investment with intermediate consumption}
\author{Oleksii Mostovyi\\
  Carnegie Mellon University,\\
  Department of  Mathematical Sciences,\\
  5000 Forbes Avenue, Pittsburgh, PA, 15213-3890, US \\
  (omostovy@andrew.cmu.edu)}

  \date{}
\maketitle

\begin{abstract}
  We consider a problem of optimal investment with
  intermediate consumption in the framework of an incomplete
  semimartingale model of a financial market.  We show that a
  {necessary} and {sufficient} condition for the validity of key
  assertions of the theory is that the value functions of the primal
  and dual problems are
  finite. % While the finiteness condition for dual
  % value function has appeared in \cite{KS2003}, the analogous
  % assumption on the primal value function is apparently new.
\end{abstract}

\section{Introduction}
\label{sec:introduction}

A fundamental problem of mathematical finance is that of an investor
who wants to invest and consume in a way that maximizes his expected
utility.  The first results for continuous time models were obtained
by Merton \cite{Merton1, Merton2} in a Markovian setting via dynamic programming arguments.  An alternative martingale
approach was developed among others by Cox and Huang \cite{CoxHuang1,
  CoxHuang2}, Karatzas, Lehoczky and Shreve \cite{KLS}, and Karatzas
and Shreve \cite{KSmmf} for complete markets and by Karatzas,
Lehoczky, Shreve and Xu \cite{KLSX}, He and Pearson \cite{HePearson1,
  HePearson2}, Kramkov and Schachermayer \cite{KS, KS2003}, Karatzas
and \v Zitkovi\'c \cite{Karatzas-Zitkovic-2003}, and \v Zitkovi\'c
\cite{Zitkovic} in an incomplete case. The main focus here was to
establish conditions under which ``key'' results, such as the
existence of primal and dual optimizers, hold.

When the consumption occurs only at maturity and the utility function
is deterministic a necessary and sufficient condition has been obtained
in Kramkov and Schachermayer \cite{KS2003}. It is stated as the
finiteness of the dual value function.  In the case of intermediate
consumption and stochastic field utility, the latest sufficient
conditions are due to Karatzas and \v Zitkovi\'c
\cite{Karatzas-Zitkovic-2003} and \v Zitkovi\'c \cite{Zitkovic}. They
are formulated in the form of several regularity assumptions such as a
uniform asymptotic elasticity.

This paper obtains necessary and sufficient conditions in the general
framework of an incomplete financial model with a stochastic field utility
and intermediate consumption occurring according to some stochastic
clock.  As in~\cite{KS2003} we assume that the dual value function is
finite (from above).  Maybe surprisingly, the only other condition we
need is the finiteness of the primal value function (from below).
Note that the latter condition holds trivially in the setting
of~\cite{KS2003}.

The remainder of the paper is organized as follows. In Section
\ref{sec:main-results} we describe the model and state the main
results. Their proofs are given in Section \ref{pfOfProp1} and are
based on the abstract versions of the main theorems presented in Section
\ref{abstractVersionOfTheMainTheorem}.

\section{Main Results}
\label{sec:main-results}

A model of a security market consists of $(d+1)$ assets: one bond and
$d$ stocks. We assume that the bond is chosen as a num\'eraire and
 denote by $S=\left( S^i\right)_{1\leq i\leq d}$
the discounted price processes of
the stocks.
 We suppose that $S$ is a semimartingale on a complete stochastic
basis $(\Omega, \mathcal{F}, \left(\mathcal{F}_t\right)_{t\in [0,
  \infty)}, \mathbb{P})$ with an infinite time horizon, $\mathcal F_0$ is the
completion of the trivial $\sigma$-algebra.

Define a portfolio $\Pi$ as a triple $(x, H, c),$ where the constant
$x$ is an initial value, $H=\left(H^i\right)_{1\leq i \leq d}$ is a
predictable $S$-integrable process of stocks' quantities, and
$c=\left(c_t\right)_{t\geq 0}$ is a nonnegative and optional process that
specifies the consumption rate in the units of the bond.

Hereafter we fix a \textit{stochastic clock} $\kappa =
\left(\kappa_t\right)_{t\geq 0}$, which is a
non-decreasing, c\'adl\'ag, adapted process such that
\begin{equation}
  \label{stochasticClock}
\kappa_0 = 0, ~~ \mathbb P\left[\kappa_{\infty}>0 \right]>0,\text{ and } \kappa_{\infty}\leq A
\end{equation}
for some  finite constant $A$. Stochastic clock represents the notion of time
according to which consumption occurs.

The discounted value process $V=\left(V_t\right)_{t\geq
  0}$ of a portfolio $\Pi$ is defined as

\begin{equation}\label{defW}
   V_t \set x + \int_0^t H_udS_u - \int_0^tc_ud\kappa_u, \quad t\geq 0.
\end{equation}
A portfolio $\Pi$ with $c\equiv 0$ is called
\textit{self-financing}.  The collection of nonnegative value
processes of self-financing portfolios with initial value $1$ is
denoted by $\mathcal X$, i.e.,
\begin{equation}\nonumber
  \mathcal{X} \set \left\{ X\geq 0:~ X_t = 1 + \int_0^tH_udS_u, ~~t\geq 0
  \right\}.
\end{equation}
A pair $(H, c),$ such that for a given $x>0$ the corresponding value
process $V$ is nonnegative, is called an \textit{$x$-admissible
  strategy}.  If for a consumption process $c$ we can find a
predictable $S$-integrable process $H$ such that $(H, c)$ is an
$x$-admissible strategy, we say that $c$ is an \textit{$x$-admissible
  consumption process}.

The set of the $x$-admissible consumption
processes corresponding to a stochastic clock $\kappa$ is denoted by $\mathcal
A(x)$, that is,
\begin{equation}
  \label{defOfAcal}
    \mathcal A(x)\set \left\{c:~c~
      \text{is $x$-admissible}\right\},~~x>0.
\end{equation}
We write $\mathcal{A} \set \mathcal{A}(1)$ for brevity.

% For $t\geq 0$ we denote by $\mathbb Q_t$ the restriction of a
% probability measure $\mathbb Q$ to $\mathcal F_t.$ A probability
% measure $ \mathbb{Q}$ is called a \textit{locally equivalent
%   martingale measure} if for every $t>0$ we have $\mathbb Q_t
% \sim\mathbb P_t$ and every $X\in\mathcal{X}$ is a local martingale under
% $\mathbb Q$.  We denote the family of locally equivalent
% martingale measures by $\mathcal{M}$ and assume that
% \begin{equation}\label{MisNotEmpty}
%   \mathcal{M} \neq \O.
% \end{equation}

The set of \textit{equivalent martingale deflators} is defined as
\begin{equation}\label{defOfYtilde}
\begin{array}{c}
  \hspace{-14mm}\mathcal{Z}\set
  \left\{Z>0:~ Z {\rm~is~ a~ c\acute{a}dl\acute{a}g~martingale,~s.t.~}Z_0=1~{\rm and}\right.\\
  \hspace{11mm}\left. XZ=(X_tZ_t)_{t\geq
      0}{\rm~is~a~local~ martingale~for~every~}X\in\mathcal X\right\}.\\
\end{array}
\end{equation}
We assume that
\begin{equation}\label{ZisNotEmpty}
  \mathcal{Z} \neq \O.
\end{equation}
This condition is closely related to the absence of arbitrage
opportunities in the sense of~\cite{Karatzas-Kardaras}.

% The
% corresponding set of c\'{a}dl\'{a}g densities is denoted by
% ${\mathcal Z}$:
% \begin{equation}\label{defOfYtilde}
%   \mathcal{Z}\set
%   \left\{Z=\left(\frac{d\mathbb{Q}_t}{d\mathbb{P}_t}\right)_{t\geq 0},
%     \quad \mathbb{Q}\in\mathcal{M}\right\}.
% \end{equation}

We now introduce an economic agent whose consumption preferences are
modeled with a \textit{utility stochastic field} $ U= U(t,\omega,x):
[0, \infty)\times\Omega\times[0, \infty)\to \mathbb{R}\cup\{-\infty\}$ satisfying the
conditions below.
% \begin{displaymath}
%   U= U(t,\omega,x):
%   [0, \infty)\times\Omega\times(0, \infty)\to \mathbb{R}.
% \end{displaymath}

\begin{Assumption}
  \label{Assumption1}
  For every $(t, \omega)\in[0, \infty)\times\Omega$ the function $x\to
  U(t, \omega, x)$ is strictly concave, increasing, continuously
  differentiable on $(0,\infty)$ and satisfies the Inada conditions:
  \begin{equation}\label{Inada}
    \lim\limits_{x\downarrow 0}U'(t, \omega, x) =
    +\infty \quad \text{and} \quad \lim\limits_{x\to
      \infty}U'(t, \omega, x) \set 0,
  \end{equation}
  where $U'$ denotes the partial derivative with respect to the third argument.
 At $x=0$ we have, by continuity, $U(t, \omega, 0) = \lim\limits_{x \downarrow 0}U(t, \omega, x)$, this value
 may be $-\infty$.
For every
  $x\geq 0$ the stochastic process $U\left( \cdot, \cdot, x \right)$ is
  optional.
\end{Assumption}

For a given initial capital $x>0$ the goal of the agent is to maximize
his expected utility. The value function of this problem is denoted by
\begin{equation}\label{primalProblem}
  u(x) \set \sup\limits_{c\in\mathcal{A}(x)} \mathbb{E}\left[
    \int_0^{\infty}U(t,\omega, c_t)d\kappa_t
  \right],\quad x>0.
\end{equation}
We use the convention
\begin{displaymath}
  \mathbb{E}\left[
    \int_0^{\infty}U(t,\omega, c_t)d\kappa_t \right] \set -\infty
  \quad \text{if} \quad \mathbb{E}\left[ \int_0^{\infty}U^{-}(t,\omega, c_t)d\kappa_t
  \right]= +\infty.
\end{displaymath}
Here and below, $W^{-}$ and $W^{+}$ denote the negative and the positive parts of a
stochastic field $W$, respectively.

Our goal is to find
conditions on the financial market and the utility field $U$ under
which the key conclusions of the utility maximization theory hold,
namely, $u$ satisfies the Inada conditions and the solution $\hat
c(x) \in\mathcal A(x)$ to~(\ref{primalProblem}) exists.

% \begin{Remark}
%   Since the consumption plans with nontrivial component singular to
%   $d\kappa$ are suboptimal, it is natural to restrict the set of
%   admissible cumulative consumptions to the ones that admit densities
%   with respect to the stochastic measure $d\kappa$ as in definition
%   (\ref{defOfAcal}).
% \end{Remark}

\begin{Remark}
For simplicity of notations we assume throughout the paper that the argument
$x$ in $U(t, \omega, x)$ represents the consumption in the discounted units,
that is, in the number of bonds. This does not restrict any
generality. Indeed, suppose that the investor's stochastic field utility is
given as $\tilde U = \tilde U(t, \omega, \tilde x)$, where the consumption $\tilde x$
is measured in the number of units of a different asset, whose discounted
value is given by a strictly positive semimartingale $A=(A_t)_{t\geq 0}$
.%, such that $A_0$ is deterministic.
Then we arrive to our framework by setting
\begin{displaymath}
 U(t, \omega, x) \set \tilde U\left(t,
    \omega, x/A_t(\omega)\right).
\end{displaymath}

%   In \eqref{primalProblem} $c$ and $U$ are measured in the discounted units.
%   This does
%   not restrict any generality.  Indeed, the case when
%   $c$ and $U$ as well as $S$, $V$, and  the bond $B=(B_t)_{t\geq 0}$ are given in the
%   units of a bank account can be reduced to the one above if we define
% the utility field $\tilde U$ as
% \begin{displaymath}
% \tilde U(t, \omega, x) \set U\left(t,
%     \omega, xB_t(\omega)\right).
% \end{displaymath}
% Note that  $B$ can be assumed to  be an arbitrary strictly positive
% semimartingale such that $B_0$ is deterministic, whereas the balance equation (\ref{defW}) in this case becomes
% \begin{equation}
%   \nonumber
%   {V_t} = {B_t} \left(\frac{x}{B_0} + \int_0^t H_ud\left(\frac{S}{B}\right)_u - \int_0^t\frac{c_u}{B_u}d\kappa_u\right), \quad t\geq 0.
% \end{equation}
\end{Remark}

To study~\eqref{primalProblem} we employ standard duality arguments as
in \cite{KS} and \cite{Zitkovic} and define the \textit{conjugate stochastic
field} $V$ to $U$ as
\begin{equation}
  \label{convexConjugate}
  V(t,\omega, y) \set \sup\limits_{x>0}\left( U(t,\omega, x) - xy \right)
  ,\quad
  \left(t, \omega, y \right) \in[0, \infty)\times\Omega\times[0,
  \infty).
\end{equation}
It is well-known that $-V$ satisfies Assumption~\ref{Assumption1}. We
also denote
\begin{equation}\label{defOfYkappa}
\begin{array}{c}
\hspace{-18mm}{\mathcal Y}(y) \set {\rm cl}\left\{Y: Y{\rm~is~c\grave adl\grave ag~adapted~and }
\right.\\
\hspace{30mm}\left. 0\leq Y\leq yZ ~\left(d\kappa\times\mathbb
    P\right){\rm~a.e.~for~some~}Z\in{\mathcal Z} \right\},
\end{array}
\end{equation}
where the closure is taken in the topology of convergence in measure $\left(d\kappa\times\mathbb
    P\right)$ on the space of real-valued optional processes. % on the stochastic  basis.
% $(\Omega, \mathcal{F}, \left(\mathcal{F}_t\right)_{t\in [0,  \infty)},
% \mathbb{P})$.
 We write ${\mathcal Y}\set{\mathcal Y}(1)$ for brevity.
%Observe that $\mathcal Y$ is the closure of the solid hull of $\mathcal Z$.

% \begin{equation}\label{defOfYkappa}
%   \begin{array}{rcl}
%     {\mathcal Y} &\set& \overline{{\mathcal Z}_{sol}},\\
%     {\mathcal Y}(y) &\set& y\mathcal{Y}, \quad y>0,
%   \end{array}
% \end{equation}
% where ${\mathcal Z}$ is defined in~\eqref{defOfYtilde} and
% for a set $B$ of c\'{a}dl\'{a}g processes $B_{sol}$ denotes its solid
% hull:
% \begin{equation}\nonumber
%   % \label{defOfYtildeSol}
%   B_{sol} \set \left\{Y': ~Y'~{\rm is~c\acute{a}dl\acute{a}g,~adapted,~and~} 0\leq
%     Y'\leq Y ~{\rm for~some~ }Y\in B \right\},
% \end{equation}
% and $\overline{B}$ stands for its closure in the topology of
% convergence in measure $\left(d\kappa\times\mathbb P\right)$ on the
% space of optional processes.

After these preparations, we define the value function of the dual optimization
problem as
\begin{equation}\label{dualProblem}
  v(y) \set \inf\limits_{Y\in{\mathcal Y}(y)} \mathbb{E}\left[
    \int_0^{\infty}V(t, \omega, Y_t )d\kappa_t
  \right],\quad y>0,
\end{equation}
where we use the convention:
\begin{displaymath}
  \mathbb{E}\left[ \int_0^{\infty}V(t, \omega, Y_t
    )d\kappa_t \right] \set +\infty
  \quad \text{if} \quad \mathbb{E}\left[ \int_0^{\infty}V^{+}(t, \omega, Y_t )d\kappa_t
  \right] = +\infty.
\end{displaymath}
%where $V^{+}$ denotes the positive part of $V$.
Theorems \ref{mainTheorem} and \ref{secondTheorem} constitute our main results.

\begin{Theorem}\label{mainTheorem}
  Assume that conditions~\eqref{stochasticClock} and (\ref{ZisNotEmpty}) and
  Assumption \ref{Assumption1} hold true and suppose
  \begin{equation}\label{mainCondition}
    v(y)<\infty~~ for~ all~y>0~~~ and~~~
    u(x) > -\infty~~ for~ all~x>0.
  \end{equation}
  Then we have:
  \begin{enumerate}
  \item $u(x)< \infty$ for all $x>0,$ $v(y)>-\infty$ for all $y>0.$
    The functions $u$ and $v$ are conjugate, i.e.,
    \begin{equation}\label{biconjugacy}
      \begin{array}{rcl}
        v(y) &=& \sup\limits_{x>0}\left(u(x) - xy\right),\quad y>0,\\
        u(x) &=& \inf\limits_{y>0}\left(v(y) + xy\right),\quad x>0.\\
      \end{array}
    \end{equation}
    The functions $u$ and $-v$ are continuously differentiable on $(0,
    \infty),$ strictly increasing, strictly concave and satisfy the
    Inada conditions:
    \begin{equation}\nonumber
      \begin{array}{lcr}
        u'(0) \set \lim\limits_{x\downarrow 0}u'(x) = +\infty, && -v'(0)
        \set \lim\limits_{y\downarrow 0}-v'(y) = +\infty,\\
        u'(\infty) \set \lim\limits_{x\to\infty}u'(x) = 0, && -v'(\infty) \set
        \lim\limits_{y\to\infty}-v'(y) = 0.\\
      \end{array}
    \end{equation}

  \item For every $x>0$ and $y>0$ the optimal solutions $\hat{c}(x)$ to
    (\ref{primalProblem}) and $\hat{Y}(y)$ to (\ref{dualProblem})
    exist and are unique. Moreover, if $y=u'(x)$ we have the dual
    relations
    \begin{displaymath}
      \hat{Y}_t(y) =  U'\left(t, \omega,  \hat{c}_t(x)
      \right),\quad t\geq 0,
    \end{displaymath}
    and
    \begin{displaymath}
      \mathbb{E}\left[ \int_0^{\infty} \hat{c}_t(x)
        \hat{Y}_t(y)d\kappa_t \right]=xy.
    \end{displaymath}

\end{enumerate}

\end{Theorem}

The finiteness conditions (\ref{mainCondition}) are
clearly necessary for the conclusions of either item 1 or 2. Notice that the condition $u(x) >-\infty$ for all $x>0$ holds
trivially if the utility stochastic field $U$ is uniformly bounded from
below by a real-valued function.

A natural question
is whether one can use the set $\mathcal Z$ instead of $\mathcal Y$ as the
dual domain and still obtain the same value function $v$.
 Theorem \ref{secondTheorem} below states that the answer is positive, however, the minimizer might lie outside of the set
$\mathcal Z$  in general, see e.g. Example 5.1 in Kramkov and Schachermayer
\cite{KS}.
Furthermore, due to a
certain \textit{symmetry} between primal and dual problems (that is explored in more
detail in Section \ref{abstractVersionOfTheMainTheorem}) a similar conclusion
is valid for the value function $u$.
Let ${\mathcal B}$ be a subset of $\mathcal A$ such that
\begin{itemize}
\item [(i)]
for every $Y\in\mathcal Y$, we have
\begin{displaymath}
\sup\limits_{c\in{\mathcal B}}\mathbb{E}\left[ \int_0^{\infty}c_tY_td\kappa_t  \right]=
\sup\limits_{c\in\mathcal{A}} \mathbb{E}\left[ \int_0^{\infty}c_tY_td\kappa_t \right],
\end{displaymath}
\item [(ii)]
the set ${\mathcal B}$
is closed under the countable convex combinations, that is, for any sequence
$\left(c^n\right)_{n\geq 1}$ of optional processes in ${\mathcal B}$ and any sequence of positive numbers $(a^n)_{n\geq 1}$ such that
$\sum_{n=1}^{\infty}a^n=1$, the process
$\sum_{n=1}^{\infty}a^nc^n$ belongs to ${\mathcal B}$.

\end{itemize}
Observe that $\mathcal Z$ is closed under the countable convex combinations.
\begin{Theorem}\label{secondTheorem}
Under the conditions of Theorem \ref{mainTheorem}, we have
    \begin{displaymath}
\begin{array}{rclc}
      v(y)& =& \inf\limits_{Z\in\mathcal{Z}}\mathbb{E}\left[ \int_0^{\infty}
        V\left(t, \omega, yZ_t\right)d\kappa_t\right],& y>0,\vspace{2mm}\\
  u(x) &= &\sup\limits_{c\in{\mathcal B}} \mathbb{E}\left[
    \int_0^{\infty}U(t,\omega, xc_t)d\kappa_t
  \right],& x>0.
\end{array}
\end{displaymath}

\end{Theorem}

The proofs of Theorems \ref{mainTheorem} and \ref{secondTheorem} will be given in
Section~\ref{pfOfProp1}  and will rely on Theorems~\ref{mainTheorem2} and
\ref{secondTheorem2}, which are the ``abstract'' versions of
Theorems~\ref{mainTheorem} and \ref{secondTheorem}, respectively. We conclude this
section with examples of the investment problems (see e.g. Karatzas \cite{Kar89} as well as Karatzas and Shreve
\cite{KSmmf}) that are included in our formulation.
Hereafter, $1_E$ denotes the indicator function of a set $E$.

\begin{Example}
 Maximization of the expected utility from
consumption:
\begin{equation}\nonumber%\label{optimalConsumption}
  u(x) = \sup\limits_{c\in\mathcal{A}(x)} \mathbb{E}\left[
    \int_0^TU(t,\omega,  c_t)dt\right].
\end{equation}
Here the clock $\kappa$ is given by
\begin{equation}\nonumber
  \kappa(t) \set \min\left(t,~T\right),~~ t\geq 0.
\end{equation}
\end{Example}

\begin{Example}Maximization of the expected utility from
consumption and terminal wealth:
\begin{equation}\label{4145}
  u(x) = \sup\limits_{c\in\mathcal{A}(x)} \mathbb{E}\left[
    \int_0^TU_1(t,\omega,  c_t)dt + U_2(\omega, c_T)\right].
\end{equation}
Here the clock $\kappa$ is given by
\begin{equation}\nonumber
  \kappa(t) \set t1_{[0, T)}(t) + (T+1)1_{[T,
    \infty)}(t),~~t\geq  0.
\end{equation}
% Problem (\ref{4145}) is a particular case of (\ref{primalProblem}), if
% we assume that the terminal wealth is consumed immediately at time
% $T.$
\end{Example}

\begin{Example}
Maximization of the expected utility from terminal wealth:
\begin{equation}\label{terminalWealth}
  u(x) = \sup\limits_{X\in \mathcal X} \mathbb{E}\left[
    U(\omega, xX_T)\right],
\end{equation}
% where the set $\mathcal X(x)$ is defined as
% \begin{equation}\nonumber
%   \mathcal X(x) \set \{ xX=\left(xX_t\right)_{t\geq 0}:~X\in\mathcal X\}, ~~x>0.
% \end{equation}
The corresponding clock process is
\begin{equation}\nonumber
  \kappa(t) \set 1_{[T, \infty)}(t),~~t\geq 0.
\end{equation}
Note that the formulation (\ref{terminalWealth}) extends the framework
of Kramkov and Schachermayer (see \cite{KS, KS2003}) to stochastic
utility.
\end{Example}

\begin{Example}
Maximization of the expected utility from
consumption over the infinite time horizon, that is
\begin{equation}\label{infiniteHorizon}
  u(x) = \sup\limits_{c\in\mathcal A(x)}\mathbb E\left[
    \int_0^{\infty} e^{-\nu t}U(t, \omega, c_t)dt
  \right], \quad x>0,~\nu >0,
\end{equation}
where the clock is defined as
\begin{displaymath}
\kappa(t)\set \int_0^t e^{-\nu s}ds = \frac{1}{\nu}\left(1 - e^{-\nu
    t}\right),~~t\geq 0.
\end{displaymath}
% Using change of variables formula we can restate it in the form
% (\ref{primalProblem}). The set $\mathcal A(x)$ is defined accordingly.
\end{Example}

\begin{Example}
Maximization of expected utility from consumption
occurring at discrete times $(t_1,\dots,t_N)$:
\begin{equation}\label{descreteTimes}
  u(x) = \sup\limits_{c\in\mathcal A(x)}\mathbb E\left[
    \sum\limits_{j=1}^NU(t_j, \omega, c_{t_j})
  \right], \quad x>0.
\end{equation}
Here the clock process is
\begin{equation}\nonumber%\label{descreteTimesClock}
\kappa(t) \set \sum\limits_{j=1}^N 1_{[t_j,+\infty)}(t),~~t\geq 0.
\end{equation}
\end{Example}

\section{Abstract versions of the main theorems}\label{abstractVersionOfTheMainTheorem}

Let $\mu$ be a finite and positive measure on a measurable space  $(\Omega,
\mathcal F)$. Denote by
$\mathbf{L}^0=\mathbf{L}^0\left(
  \Omega, \mathcal F, \mu\right)$ the vector space of (equivalence classes of)
real-valued measurable functions on
$(\Omega, \mathcal F, \mu)$ topologized by convergence in measure $\mu$.
Let $\mathbf{L}^0_{+}$ denote its positive
orthant, i.e.,
\begin{displaymath}
\mathbf{L}^0_{+} = \left\{ \xi\in\mathbf{L}^0\left(
  \Omega, \mathcal F, \mu\right):~\xi \geq 0\right\}.
\end{displaymath}
For any $\xi$ and $\eta$ in $\mathbf{L}^0$ we write
\begin{displaymath}
\langle \xi, \eta \rangle \set \int_{\Omega} \xi \eta d\mu,
\end{displaymath}
whenever the latter integral is well-defined. Let $\mathcal{C}, \mathcal{D}$ be  subsets of
$\mathbf{L}^0_{+}$ that satisfy the conditions below.

\begin{enumerate}
\item We have
  \begin{equation}\label{polarityCD}
    \begin{array}{rcl}
      \xi\in\mathcal{C} &\Leftrightarrow& \langle \xi, \eta \rangle \leq 1~{\rm
        for~all~}\eta \in\mathcal{D},\\
      \eta\in\mathcal{D} &\Leftrightarrow& \langle \xi, \eta\rangle \leq 1~{\rm
        for~all~}\xi \in\mathcal{C}.\\
    \end{array}
  \end{equation}

\item $\mathcal{C}$ and $\mathcal{D}$ contain at least one
  strictly positive element:
%   \begin{equation}\label{positivityCD}
%     \begin{array}{rcl}
%       {\rm there~exists~}&\xi\in\mathcal{C}& {\rm~such~that~}\xi>0,\\
%       {\rm there~exists~}&\eta\in\mathcal{D}& {\rm~such~that~}\eta>0.\\
%     \end{array}
%   \end{equation}
\begin{equation}\label{positivityCD}
{\rm there~are~~}\xi\in\mathcal{C},~ \eta\in\mathcal{D} {\rm~~such~that~~}\min(\xi,~\eta)>0~~\mu~a.e.
\end{equation}
\end{enumerate}
Observe that our construction of the abstract sets $\mathcal C$ and $\mathcal D$ is
similar to the one in \cite{KS}, however we
do not require a constant to be an element of $\mathcal C$. This leads to a \textit{symmetry} between the sets $\mathcal C$ and $\mathcal D$ that
plays an important role in the proofs. Also notice
that $\mathcal C$ and $\mathcal D$ are convex and bounded in $\mathbf{L}^0\left(\mu\right)$.
For $x>0$ and $y>0$ we define the sets:
\begin{equation}\label{defCxDy}
  \begin{array}{rcrcl}
    \mathcal{C}(x) &\set &x\mathcal{C} &\set&\left\{ x\xi: ~\xi\in\mathcal{C}\right\}, \\
    \mathcal{D}(y) &\set &y\mathcal{D} &\set&\left\{ y\eta: ~\eta\in\mathcal{D}\right\}. \\
  \end{array}
\end{equation}

% \subsection{Utility Field}
Consider a \textit{stochastic utility function}
%$\mathcal{F}\otimes\mathcal{B}\left([0, \infty) \right)$ measurable function
$U$: $\Omega\times[0, \infty)\to \mathbb{R}\cup\{-\infty\}$, which satisfies
the following conditions.
\begin{Assumption}\label{Assumption2}
  For every $ \omega\in \Omega$ the function $x\to U( \omega, x)$ is
  strictly concave, increasing, continuously differentiable on $(0,\infty)$, and
  satisfies the Inada conditions:
  \begin{equation}\label{InadaAbstract}
    \lim\limits_{x\downarrow 0}U'(\omega, x) =
    +\infty \hspace{3mm}{\rm and}\hspace{3mm} \lim\limits_{x\to
      \infty}U'(\omega, x) = 0,
  \end{equation}
 \looseness+1
where $U'(\cdot, \cdot)$ denotes the partial derivative with respect
  to the second argument. At $x=0$ we have, by continuity, $U(\omega, 0) = \lim\limits_{x\downarrow
    0}U(\omega, x)$, this value may be $-\infty$. For every $x\geq0$ the function $U\left(
    \cdot, x \right)$ is measurable.
\end{Assumption}
Define the \textit{conjugate function} $V$ to $U$ as
\begin{equation}\nonumber
  V(\omega, y) \set \sup\limits_{x>0}\left( U(\omega, x) - xy \right)
  ,\quad \left(\omega, y \right) \in \Omega\times[0,
  \infty).
\end{equation}
Observe that $-V$ satisfies Assumption \ref{Assumption2}.
% Define the \textit{inverse marginal utility field} $I: \Omega\times(0,
% \infty)\to \mathbb{R}:$
% \begin{equation}\nonumber
%   I\left(\omega, y\right) \set \left(U'\left(\omega,
%       y\right)\right)^{-1},\hspace{3mm}{\rm for \hspace{1mm}
%     all\hspace{1mm}}\left(\omega, y \right) \in\Omega\times(0,
%   \infty).
% \end{equation}
For a function  $W$ on $\Omega\times[0,
\infty)$ and a function $\xi\in\mathbf L^0_{+}$  we will write $W(\xi) \set
W(\omega, \xi(\omega))$. Recall that
$W^{+}$ and $W^{-}$ denote the positive and the negative parts of $W$, respectively.

% \subsection{Optimization problems}

Now we can state the optimization problems:
\begin{equation}\label{primalProblem2}
  u(x) = \sup\limits_{\xi\in\mathcal{C}(x)}\int_{\Omega} {U}(\xi)d\mu, \hspace{5mm}x>0,
\end{equation}
\begin{equation}\label{dualProblem2}
  v(y) = \inf\limits_{\eta\in\mathcal{D}(y)}\int_{\Omega} {V}(\eta)d\mu, \hspace{5mm}y>0,
\end{equation}
where we used the convention:
% To ensure that the integrals in (\ref{primalProblem2}) and
% (\ref{dualProblem2}) are well defined, we set the convention:
\begin{displaymath}
\begin{array}{rcl}
\int_{\Omega} U(\xi)d\mu \set -\infty&
 {\rm if} &
\int_{\Omega} U^{-}(\xi)d\mu = +\infty,\\
\int_{\Omega} V(\eta)d\mu \set +\infty & {\rm if}&
\int_{\Omega} V^{+}(\eta)d\mu = +\infty.\\
\end{array}
\end{displaymath}
The following theorem is an abstract version of Theorem \ref{mainTheorem}.

% In order to make notations more convenient we define the functionals
% $\mathsf{U},$ $\mathsf{V},$ $\mathsf{U}',$ and $\mathsf{I}$ on the set
% of strictly positive measurable functions that take values in the set
% of measurable functions:
% \begin{displaymath}
% \begin{array}{rl}
%   \mathsf{U}(c)(\omega) \set U\left( \omega, c(\omega)\right),&
% \mathsf{V}(h)(\omega) \set V\left(\omega, h(\omega)\right),\\
%  \mathsf{U}'(h)(\omega) \set U'\left(\omega, h(\omega)\right),&
%   \mathsf{I}(h)(\omega) \set I\left( t, \omega, h(\omega)\right). \\
% \end{array}
% \end{displaymath}

% \subsection{Abstract version of theorem \ref{mainTheorem}}

% The value functions $u,$ $-v$ are clearly concave. Define by $u'$ and
% $v'$ the right-continuous versions of the derivatives of $u$ and $v.$

\begin{Theorem}\label{mainTheorem2}
  Assume that $\mathcal{C}$ and $\mathcal{D}$ satisfy conditions
  (\ref{polarityCD}) and (\ref{positivityCD}). Let Assumption
  \ref{Assumption2} hold and suppose
  \begin{equation}\label{mainCondition2}
    v(y)<\infty~~for~ all~y>0\quad and\quad u(x) > -\infty~~for~ all~x>0.
  \end{equation}
  Then we have:
  \begin{enumerate}
  \item $u(x)< \infty$ for all $x>0,$ $v(y)>-\infty$ for all $y>0.$
    The functions $u$ and $v$ satisfy the biconjugacy relations, i.e.,
    \begin{equation}\label{biconjugacy2}
      \begin{array}{rcl}
        v(y) &=& \sup\limits_{x>0}\left(u(x) - xy\right),\quad y>0,\\
        u(x) &=& \inf\limits_{y>0}\left(v(y) + xy\right),\quad x>0.\\
      \end{array}
    \end{equation}
    The functions $u$ and $-v$ are continuously differentiable on $(0,
    \infty)$, strictly increasing, strictly concave, and satisfy the Inada
    conditions:
    \begin{equation}\nonumber
      \begin{array}{lcr}
        u'(0) \set \lim\limits_{x\downarrow 0}u'(x) = +\infty, && -v'(0) \set \lim\limits_{y\downarrow 0}-v'(y) = +\infty,\\
        u'(\infty) \set \lim\limits_{x\to\infty}u'(x) = 0, && -v'(\infty)
        \set
        \lim\limits_{y\to\infty}-v'(y) = 0.\\
      \end{array}
    \end{equation}

  \item For every $x>0$ the optimal solution $\hat{\xi}(x)$ to (\ref{primalProblem2})
    exists  and is
    unique. For every $y>0$ the optimal solution $\hat{\eta}(y)$ to
    (\ref{dualProblem2})  exists and is unique. If $y=u'(x)$, we have the dual
    relations
    \begin{displaymath}
        \hat{\eta}(y) = {U}'\left(\hat{\xi}(x) \right)
        %,\quad \hat{\xi}(x)= {I}\left(\hat{h}(\eta)\right),
        \quad\mu {\rm~a.e.}
    \end{displaymath}
and
    \begin{displaymath}
      \langle \hat{\xi}(x),\hat{\eta}(y)\rangle = xy.
      \end{displaymath}

 \end{enumerate}% \item
%%%%%%%%%%%%% item 3 %%%%%%%%%%%%%%%%%%%%%%%%%%%%%%%%%%%%%%%

\end{Theorem}
In order to state an abstract version of Theorem \ref{secondTheorem} we need
the following definitions. Let
$\tilde{\mathcal{D}}$ be a subset of $\mathcal{D}$
such that
\begin{itemize}%[(i)]
\item [(i)]
$\tilde{\mathcal{D}}$ is closed under the countable convex
combinations,
% i.e., for every sequence $\left(
%   \eta^n\right)_{n\geq 1}\subset \tilde{\mathcal{D}}$ and every
% sequence of positive numbers $\left( a^n \right)_{n\geq 1}$ such
% that $\sum\limits_{n=1}^{\infty}a^n=1$ the function
% $\sum\limits_{n=1}^{\infty}a^n \eta^n$ belongs to
% $\tilde{\mathcal{D}}$,

\item [(ii)]
for every $\xi\in\mathcal{C}$ we have
\end{itemize}
\begin{equation}\label{11191}
  \sup\limits_{\eta\in\mathcal{D} }\langle \xi,\eta \rangle =
  \sup\limits_{\eta\in\tilde{\mathcal{D}} }\langle \xi,\eta \rangle.
\end{equation}
Likewise, define $\tilde{\mathcal{C}}$ to be a subset of $\mathcal{C}$ such that
\begin{itemize}
\item [(iii)]
$\tilde{\mathcal{C}}$ is closed under the countable convex
combinations,
% i.e., for every sequence $\left(
%   \xi^n\right)_{n\geq 1}\subset \tilde{\mathcal{C}}$ and every
% sequence of positive numbers $\left( a^n \right)_{n\geq 1}$ such
% that $\sum\limits_{n=1}^{\infty}a^n=1$ the function
% $\sum\limits_{n=1}^{\infty}a^n \xi^n$ belongs to
% $\tilde{\mathcal{C}}$,

\item [(iv)]
for every $\eta\in\mathcal{D}$ we have
\begin{equation}\nonumber%\label{}
  \sup\limits_{\xi\in\mathcal{C} }\langle \xi,\eta \rangle =
  \sup\limits_{\xi\in\tilde{\mathcal{C}} }\langle \xi,\eta \rangle.
\end{equation}
\end{itemize}

\begin{Theorem}\label{secondTheorem2}
Under the conditions of Theorem \ref{mainTheorem2}, we have
    \begin{displaymath}
      \begin{array}{rclc}
      v(y) & =& \inf\limits_{\eta\in\mathcal{\tilde{D}} }\int_{\Omega}
      {V}\left(y\eta\right)d\mu,& \quad y>0.
      \vspace{2mm}\\
      u(x) &=& \sup\limits_{\xi\in\mathcal{\tilde{C}} }\int_{\Omega}
      {U}\left(x\xi\right)d\mu,& \quad x>0.\\
      \end{array}
    \end{displaymath}

\end{Theorem}

The proofs of Theorem \ref{mainTheorem2} and \ref{secondTheorem2} are given via several lemmas.
%\subsection{Existence of a minimizer to the dual problem}

\begin{Lemma}\label{sub-conjugacy}
  % Assume that $\mathcal{C}$ and $\mathcal{D}$ satisfy conditions
  % (\ref{polarityCD}) and (\ref{positivityCD}). Then
  Under the conditions of Theorem \ref{mainTheorem2},
we have %$u$ is a  sub-conjugate of $v,$ that is
  \begin{equation}\label{subconjugacyEqn}
    v(y) \geq \sup\limits_{x>0}\left(u(x) - xy\right),~~~~y>0.
  \end{equation}
  As a result, both $u$ and $v$ are real-valued functions, such that
\begin{displaymath}
\limsup\limits_{x\to\infty}\frac{u(x)}{x}\leq 0 \quad{\rm and}\quad \liminf\limits_{y\to\infty}\frac{v(y)}{y}\geq 0.
\end{displaymath}
\end{Lemma}
\begin{proof} Fix $x>0$ and $y>0.$ We have
\begin{equation}\label{5-27-1}
  \sup\limits_{\xi\in\mathcal C(x)}\inf\limits_{\eta\in\mathcal
    D(y)}\int_{\Omega}\left(  U(\xi) - \xi\eta \right)d\mu \leq \inf\limits_{\eta\in\mathcal
    D(y)}\sup\limits_{\xi\in\mathcal C(x)}\int_{\Omega}\left(  U(\xi) - \xi\eta \right)d\mu.
\end{equation}
Using (\ref{polarityCD}) we can bound the left-hand side from below by
$u(x) - xy$:
\begin{equation}\nonumber
  \begin{array}{c}
    \sup\limits_{\xi\in\mathcal C(x)}\inf\limits_{\eta\in\mathcal
      D(y)}\int_{\Omega}\left( U(\xi) - \xi\eta \right)d\mu
    \geq \sup\limits_{\xi\in\mathcal C(x)}\left(\int_{\Omega}  U(\xi)d\mu - xy
    \right) = u(x) - xy. \\
  \end{array}
\end{equation}
Since $ V(\eta)\geq  U(\xi) - \xi\eta$ for every $\xi \geq 0$ and $\eta \geq 0$, we can
bound the right-hand side of (\ref{5-27-1}) from above by
$v(y)$:
\begin{equation}\nonumber
  \inf\limits_{\eta\in\mathcal
    D(y)}\sup\limits_{\xi\in\mathcal C(x)}\int_{\Omega}\left( U(\xi) - \xi\eta
  \right)d\mu \leq \inf\limits_{\eta\in\mathcal
    D(y)}\int_{\Omega} V(\eta)
  d\mu = v(y),
\end{equation}
and the result follows.
% This proves (\ref{subconjugacyEqn}). In
% turn, (\ref{subconjugacyEqn}) together with (\ref{mainCondition2})
% imply the remaining assertions of the lemma.
\end{proof}

The techniques in Kramkov and Schachermayer \cite{KS2003} inspired the proof
of the following lemma.
\begin{Lemma}\label{uiOfVminus}
  Under the conditions of Theorem \ref{mainTheorem2},
  for every $y>0$ the family $\left( {V}^{-}\left(h\right)\right)_{h\in\mathcal{D}(y)}$
  is uniformly integrable.
\end{Lemma}
\begin{proof}
% It follows from Lemma \ref{sub-conjugacy} that
% \begin{equation}\label{1282}
%   \liminf\limits_{y\to\infty}\frac{v(y)}{y} \geq 0.
% \end{equation}
Fix $y>0.$ Assume by contradiction that $\left(
  {V}^{-}\left(h\right)\right)_{h\in\mathcal{D}(y)}$ is not a
uniformly integrable family. Then we can find a sequence $\left( \eta^n
\right)_{n\geq 2}\subset\mathcal{D}(y)$, a sequence $\left(A^n\right)_{n\geq
  2}$ of disjoint subsets of $\left(\Omega, \mathcal{F}\right)$ and a
constant $\alpha>0$ such that
\begin{equation}\nonumber
  \int_{\Omega} {V}^{-}\left(\eta^n\right)1_{A^n}d\mu \geq
  \alpha,\hspace{3mm}n\geq 2.
\end{equation}
Since $v(y)<\infty,$ there exists $\eta^1\in\mathcal D(y)$ such that
\begin{equation}\nonumber
  M\set \int_{\Omega} {V}^{+}\left( \eta^1\right)d\mu  < \infty.
\end{equation}
Define a sequence $\left( \zeta^n\right)_{n\geq 1}$ as
$\zeta^n \set \sum\limits_{k=1}^n \eta^k$, $n\geq 1$. Then by
(\ref{polarityCD})  for every $\xi\in\mathcal{C}$ we have
\begin{displaymath}
\langle\zeta^n, \xi\rangle = \sum\limits_{k=1}^n\langle \eta^k,
\xi\rangle \leq ny.
\end{displaymath}
Thus $\zeta^n\in \mathcal{D}(ny)$, $n\geq 1.$ Now, since ${V}^{-}$ is nonnegative and
nondecreasing we get
\begin{equation}\nonumber
  \begin{array}{rcl}
    \int_{\Omega} {V}^{-}\left( \zeta^n\right)d\mu &\geq&
    \int_{\Omega}{\sum\limits_{k=2}^n {V}^{-}\left(\sum\limits_{j=1}^n\eta^j\right)1_{A^k}}d\mu\\
    &\geq&\int_{\Omega}\sum\limits_{k=2}^n {V}^{-}\left(\eta^k\right)1_{A^k}d\mu\\
    &\geq& \alpha (n-1),\quad n\geq 2.\\
  \end{array}
\end{equation}
On the other hand, since ${V}^{+}$ is nonincreasing  we obtain
\begin{equation}\nonumber
  \int_{\Omega} {V}^{+}\left(\zeta^n\right)d\mu \leq
  \int_{\Omega} {V}^{+}\left(\eta^1\right)d\mu = M < \infty.
\end{equation}
Therefore we deduce that
\begin{equation}\nonumber
  \int_{\Omega} {V}\left(\zeta^n\right)d\mu\leq M-\alpha (n-1),\quad n\geq 2.
\end{equation}
Consequently,
\begin{equation}\nonumber
  \liminf\limits_{ z\to\infty}\frac{v(z)}{z}\leq
  \liminf\limits_{n\to\infty}\frac{\int_{\Omega}{V}\left(\zeta^n\right)d\mu}{ny} \leq
  \liminf\limits_{n\to\infty}\frac{M-\alpha (n-1)}{ny} = -\frac{\alpha}{y} < 0,
\end{equation}
which contradicts to the conclusion of Lemma \ref{sub-conjugacy}.
\end{proof}

We need a version of Koml\'os' lemma for the set
$\mathcal D.$ Some other formulations of Koml\'{o}s' lemma are proven
in \cite{Komlos-1967, DS, BieglBock-Schachermayer-Veliyev-2010,
  Schwartz-1986}.
\begin{Lemma}\label{Komlos}
Assume that the sets $\mathcal C$ and $\mathcal{D}$
  satisfy (\ref{polarityCD}) and (\ref{positivityCD}).
  Let  $\left(\eta^n\right)_{n\geq
    1}\subset\mathcal{D}$.  Then
  there exists a sequence of convex combinations  $\zeta^n\in{\rm
    conv}\left(\eta^n, \eta^{n+1},\dots\right),$ $n\geq 1,$ and an element
 $\hat{\eta}\in\mathcal{D},$ such that $\left(\zeta^n\right)_{n\geq 1}$
 converges $\mu$ a.e. to $\hat {\eta}$.
\end{Lemma}
\begin{proof} Using Lemma A1.1 p.515 in \cite{DS} we can construct a sequence
$\zeta^n\in{\rm conv}\left(\eta^n, \eta^{n+1},\dots\right),$ $n\geq 1,$
such that $\left(\zeta^n\right)_{n\geq 1}$ converges $\mu$ a.e. to an element $\hat \eta$.
 By convexity of the set $ \mathcal{D}$ we obtain that
$\left( \zeta^n\right)_{n\geq 1}$ is a subset of $\mathcal{D}.$ By Fatou's
lemma for every $\xi\in\mathcal C$ we have
\begin{displaymath}
  \langle  \xi, \hat{\eta}\rangle \leq \liminf\limits_{n\to\infty}\langle \xi,\zeta^n\rangle \leq 1.
\end{displaymath}
Hence, $\hat{ \eta} \in\mathcal D$.
\end{proof}

% \subsection{$v(y) = \mathbb{E} \left[ V(\hat{h}(y))\right]$}
%%%%%%%%%%%%%%%%%%%%%%%%%%%%%
\begin{Lemma}\label{ExistenceOfDualOptimizer}
  Under conditions of Theorem \ref{mainTheorem2} for each $y>0$ there
  exists a unique $\hat{\eta}(y)\in\mathcal{D}(y)$, such that
  \begin{equation}\label{11-25-2011-1}
    v(y) = \int_{\Omega}{V}\left(\hat{\eta}(y)\right)d\mu.
  \end{equation}
  As a consequence, $v$ is strictly convex.
\end{Lemma}

\begin{proof} Fix $y>0.$ Let $\left(\eta^n\right)_{n=1}^{\infty}\subset\mathcal{D}(y)$ be
a minimizing sequence, i.e.,
\begin{equation}\nonumber
  v(y) = \lim\limits_{n\to\infty}\int_{\Omega} {V}\left(\eta^n\right)d\mu.
\end{equation}
It follows from Lemma \ref{Komlos} that there exists a sequence of
convex combinations $\zeta^n\in{\rm conv}\left( \eta^n, \eta^{n+1},
  \dots\right)$, $n\geq 1$, and an
element $\hat{\eta}(y)\in \mathcal{D}(y),$ such that
$\left(\zeta^n\right)_{n=1}^{\infty}$ converges $\mu$ a.e. to $\hat{\eta}(y)$.

Using convexity of $V$, Lemma \ref{uiOfVminus}, and Fatou's lemma we get
\begin{equation}\nonumber
  v(y)=\liminf\limits_{n\to\infty}\int_{\Omega}{V}\left(\eta^n\right)d\mu
  \geq \liminf\limits_{n\to\infty}\int_{\Omega}{V}\left(\zeta^n\right)d\mu
  \geq\int_{\Omega}{V}\left(\hat{\eta}(y)\right)d\mu.
\end{equation}
Therefore (\ref{11-25-2011-1}) holds.
Uniqueness of the minimizer to (\ref{dualProblem2})
follows from the strict convexity of ${V}$.

 To show the strict convexity of $v$ fix $y_1<y_2.$
Since $\frac{\hat{\eta}({y_1})+\hat{\eta}({y_2})}{2}\in\mathcal{D}\left(
  \frac{y_1+y_2}{2}\right)$ and ${V}$ is strictly convex we
obtain
\begin{equation}\nonumber
  v\left(\frac{y_1+y_2}{2}\right)\leq \int_{\Omega}{V}\left(
    \frac{\hat{\eta}({y_1})+\hat{\eta}({y_2})}{2}\right)d\mu<
  \frac{v(y_1) + v(y_2)}{2}.
\end{equation}
\end{proof}

 By the symmetry between the optimization problems (\ref{primalProblem2}) and
 (\ref{dualProblem2}), the following result is a corollary to
Lemma \ref{ExistenceOfDualOptimizer}.

\begin{Lemma}\label{existenceOfPrimalOptimizer}
  Under the assumptions of Theorem \ref{mainTheorem2}, for every $x>0$
  there exists a unique maximizer to the primal problem
  (\ref{primalProblem2}).  As a consequence, $u$ is strictly concave.
\end{Lemma}

%\subsection{Conjugacy of $u$ and $v$}\label{subsection-conjugacy}

\begin{Lemma}\label{conjugacy}Under the assumptions of Theorem
  \ref{mainTheorem2}, we have
\begin{equation}\label{58121}
v(y)=\sup\limits_{x>0}\left(u(x) - xy\right),\quad y>0.
\end{equation}
\end{Lemma}
\begin{proof}
The two-step proof is based on the change of num\'eraire ideas.

\underline{\textit{Step 1.}}
~~Let us show (\ref{58121}) assuming that
\begin{displaymath}
{\rm the~constant~function}~~1\in\mathcal C\quad{\rm and}\quad\int_{\Omega}U(1)d\mu>-\infty.
\end{displaymath}
In this case $\int_{\Omega}U(x)d\mu$ is finite for any constant $x\geq 1$.
Let $\mathcal S_n$ be the
set of all nonnegative, measurable functions $\xi: \Omega\to [0, n],$
i.e.,
\begin{equation}\label{defS_n}
  \mathcal S_n\set\left\{\xi\in \mathbf{L}^0: \xi(\omega)\in[0, n]{\hspace{1mm}\rm
      for\hspace{1mm} all}\hspace{1mm} \omega\in
    \Omega \right\},\quad n>0.
\end{equation}
% $\mathcal S_n$ is a closed subset of a ball in
% $\mathbf{L}^{\infty}(\mu).$ By Banach-Alaoglu's theorem $\mathcal S_n$ is
% $\sigma(\mathbf{L}^{\infty},\mathbf{L}^1)$ compact.
The sets $\mathcal S_n$ are $\sigma(\mathbf{L}^{\infty},\mathbf{L}^1)$ compact. Fix
$y>0$.
Since $\mathcal{D}(y)$ is convex and $ {U}$ is concave, the minimax theorem (see \cite{Str}, Theorem 45.8)
gives the following equality
\begin{equation}\label{06131}
  \sup\limits_{\xi\in \mathcal S_n}\inf\limits_{{\eta}\in\mathcal{D}(y)}
  \int_{\Omega} \left( {U}(  \xi) -  \xi {\eta} \right)d\mu
  =\inf\limits_{{\eta}\in\mathcal{D}(y)}\sup\limits_{\xi\in \mathcal S_n}
  \int_{\Omega}\left( {U}( \xi) -  \xi {\eta}
  \right)d\mu.
\end{equation}
Denote
\begin{displaymath}
  \mathcal{C}'(x) \set \left\{
    \xi\in \mathcal{C}(x):
    \sup\limits_{{\eta}\in\mathcal{D}(y)}\langle \xi,{\eta}\rangle = xy
  \right\}.
\end{displaymath}
It follows from (\ref{defCxDy}) that
$\bigcup\limits_{x>0}\mathcal{C}'(x) \bigcup\left\{ \xi\equiv 0 \right\} =
\bigcup\limits_{x>0}\mathcal{C}(x)$. As a result, we get
\begin{equation}\label{06132}
  \begin{array}{rcl}
    \sup\limits_{x>0}\left(u(x) - xy\right)
    % & =& \sup\limits_{x>0}\sup\limits_{\xi\in\mathcal
    %   C(x)}\left(\int_{\Omega}  U(\xi)d\mu
    %   -xy\right)\\
    &=& \sup\limits_{x>0}\sup\limits_{\xi\in\mathcal
      C'(x)}\left(\int_{\Omega}  U(\xi)d\mu
      -xy\right) \\
    % &=& \sup\limits_{x>0}\sup\limits_{\xi\in\mathcal
    %   C'(x)}\left(\int_{\Omega}  U(\xi)d\mu
    %   -\sup\limits_{{\eta}\in\mathcal D(y)}\langle \xi, {\eta} \rangle\right) \\
    % &=& \lim\limits_{n\to\infty}\sup\limits_{\xi\in
    %   \mathcal S_n}\left(\int_{\Omega}  U(\xi)d\mu
    %   -\sup\limits_{{\eta}\in\mathcal D(y)}\langle \xi, {\eta} \rangle\right) \\
    &\geq&\lim\limits_{n\to\infty}\sup\limits_{\xi\in \mathcal S_n}\inf\limits_{{\eta}\in\mathcal{D}(y)}\int_{\Omega}\left( {U}( \xi) -
       \xi {\eta}\right)d\mu.\\
  \end{array}
\end{equation}
%where the latter expression is finite by the choice of $ $.
In view of (\ref{06131}), (\ref{06132}), and Lemma \ref{sub-conjugacy}
it suffices to show that
\begin{equation}\label{5-27-6}
  v(y)=\lim\limits_{n\to\infty}\inf\limits_{{\eta}\in\mathcal{D}(y)}\sup\limits_{\xi\in \mathcal S_n}
  \int_{\Omega} \left( {U}( \xi) -  \xi {\eta}\right)d\mu.
\end{equation}
For each $n\geq 1$ define $V^n$ as follows:
\begin{displaymath}
    V^n(z) \set \sup\limits_{0<x\leq n}\left( U(  x) -  xz\right),\quad z>0.
\end{displaymath}
Then via pointwise maximization we get
\begin{displaymath}
     \inf\limits_{{\eta}\in\mathcal{D}(y)}\sup\limits_{\xi\in \mathcal S_n}
  \int_{\Omega} \left( {U}( \xi) -  \xi {\eta}\right)d\mu
=
 \inf\limits_{{\eta}\in\mathcal{D}(y)}
  \int_{\Omega}  {V}^n({\eta})d\mu
 \set v^n(y).
\end{displaymath}
Notice that $v^n\leq v$ and $\left(v^n(y)\right)_{n\geq 1}$ is an increasing sequence.
Let $\left({\eta}^n\right)_{n\geq 1}\subset \mathcal{D}(y)$ be such that
\begin{equation}\label{4-6-1}
  \lim\limits_{n\to\infty}v^n(y) =
  \lim\limits_{n\to\infty}\int_{\Omega}  {V}^n({\eta}^n) d\mu.
\end{equation}
It follows from Lemma \ref{Komlos}, that there exists a sequence ${\zeta}^n\in
{\rm conv}({\eta}^n, {\eta}^{n+1}, \dots),$ $n\geq 1,$ such that
$\left({\zeta}^n\right)_{n\geq 1}$ converges  $\mu$ a.e. to
a function $\hat{{\zeta}}\in\mathcal{D}(y)$.

We claim that $\left(
V^n\right)^{-}({\zeta}^n),$ $n\geq 2,$ is a uniformly integrable sequence.
% On $\{{\zeta}^n\geq  U'( n) \}$ we have $\left( V^n\right)^{-}({\zeta}^n) = V^{-}({\zeta}^n)$. Consequently uniform integrability of $\left(
%   V^n\right)^{-}({\zeta}^n)1_{\{{\zeta}^n\geq U'( n) \}},$ $n\geq 2$,
% follows from Lemma \ref{uiOfVminus}.
%  Whereas on $\{{\zeta}^n <  U'(n) \}$ we get
% \begin{displaymath}
%   \left(  V^n\right)^{-}({\zeta}^n)1_{\{{\zeta}^n <  U'( n)
%     \}} \leq U^{-}(2 ) + 2 U'(2 ),\quad n\geq 2.
% \end{displaymath}
% From monotonicity of $U'$ we deduce
% \begin{displaymath}
%  U'(2 ) \leq \int_{1}^2 U'(x)dx = U(2)
% - U(1),
% \end{displaymath}
% where $\left(U(2)- U(1)\right)$ is integrable.
% Consequently $\left( V^n\right)^{-}({\zeta}^n)1_{\{{\zeta}^n <
%   U'( n) \}},$ $n\geq 2,$ is dominated from above by an integrable
% function. As a result, the sequence $\left( V^n\right)^{-}({\zeta}^n)$, $n\geq 2$, is uniformly integrable.
Indeed, for $n\geq 2$ we have
\begin{displaymath}
V^n(\zeta) \geq V^2(\zeta) \geq V(\zeta)1_{\{\zeta\geq U'(2)\}} + \left(U(2) - 2U'(2)\right)1_{\{\zeta< U'(2)\}}.
\end{displaymath}
The concavity of $U$ yields that $U'(2) \leq U(2) - U(1)$. Therefore,
\begin{displaymath}
V^n(\zeta) \geq \min{\left( V(\zeta),~ 2U(1) - U(2) \right)},\quad n\geq 2.
\end{displaymath}
The uniform integrability of $\left( V^n\right)^{-}({\zeta}^n)$, $n\geq 2$,
follows now from Lemma \ref{uiOfVminus} and the integrability of $U(1)$ and $U(2)$.

Therefore from the convexity of
$ {V}^n$ and Fatou's lemma we get
\begin{displaymath}\nonumber
    \lim\limits_{n\to\infty}  \int_{\Omega}  {V}^n({\eta}^n)d\mu
    \geq\liminf\limits_{n\to\infty}  \int_{\Omega}  {V}^n({\zeta}^n) d\mu
    \geq\int_{\Omega} {V}(\hat{{\zeta}}) d\mu\geq v(y),
\end{displaymath}
which in view of (\ref{4-6-1}) implies (\ref{5-27-6}).
%This concludes the proof the lemma.

\underline{\textit{Step 2.}}~~
Here we show how the general case can be reduces to the one in Step 1. Let $\hat{\xi} \set \argmin\limits_{\xi\in\mathcal C(1/2)}\int_{\Omega}
U(\xi)d\mu$ and $\xi_0$ be a strictly positive element of $\mathcal C(1/2)$. Both
$\hat{\xi}$ and $\xi_0$ exist by Lemma \ref{existenceOfPrimalOptimizer} and
assumption (\ref{positivityCD}) respectively. Define
\begin{displaymath}
\zeta \set \max (\hat{\xi}, \xi_0).
\end{displaymath}
Then $\zeta \in\mathcal C$ and $\int_{\Omega}U(\zeta)d\mu$ is finite.
Let
\begin{displaymath}
\begin{array}{rcl}
\tilde U(x) &\set& U(\zeta x),\\
\mathcal{\tilde C}(x) &\set& \left\{ \xi:~\xi\zeta\in\mathcal
  C(x)\right\},
\end{array}
\end{displaymath}
then
\begin{displaymath}
u(x) = \sup\limits_{\xi\in\mathcal{\tilde C}(x)}\int_{\Omega}\tilde
U(\xi)d\mu,\quad x>0.
\end{displaymath}
Similarly, define
\begin{displaymath}
\begin{array}{rcl}
\tilde V(y) &\set& V\left(y/\zeta\right),\\
\mathcal{\tilde D}(y) &\set& \left\{ \eta:~\eta/\zeta\in\mathcal D(y)\right\},
\end{array}
\end{displaymath}
then we have
\begin{displaymath}
v(y) = \inf\limits_{\eta\in\mathcal{\tilde D}(y)}\int_{\Omega}\tilde
V(\eta)d\mu,\quad y>0.
\end{displaymath}
Observe that $\tilde U$ satisfies assumption \ref{Assumption2}, $\tilde V$ is
the conjugate function to $\tilde U$, whereas the sets $\mathcal{\tilde C}(1)$
and $\mathcal{\tilde D}(1)$ satisfy the bipolar relations (\ref{polarityCD})
and (\ref{positivityCD}). Moreover,
\begin{displaymath}
1\in\mathcal{\tilde C}(1)\quad{\rm and}\quad\int_{\Omega}\tilde U(1)d\mu>-\infty.
\end{displaymath}
Now (\ref{58121}) follows from Step 1.

\end{proof}

%\subsection{Proof of Theorem \ref{mainTheorem2}}
\begin{proof}[Proof of Theorem \ref{mainTheorem2}]
Observe that by Lemmas \ref{existenceOfPrimalOptimizer} and
\ref{ExistenceOfDualOptimizer} both functions $u$ and $-v$ are
strictly concave. Thus, conjugacy relations (\ref{biconjugacy2})
follow from Lemma \ref{conjugacy} and Theorem 12.2 in Rockafellar
\cite{Rok} (if we extend $u$ by the value $-\infty$ on $(-\infty,0]$). In
turn, the strict concavity of $u$ and $-v$, (\ref{biconjugacy2}), and
Theorem 26.3 in \cite{Rok} imply differentiability of $u$ and $v$
everywhere in their domains.
% Now since $u$ and
% $-v$ are increasing, Inada conditions for them follow from
% Lemma \ref{sub-conjugacy}.

Fix $x>0$ and take $y=u'(x).$ Let $\hat{\eta} \in \mathcal{D}(y)$ be the
optimizer to the dual problem (\ref{dualProblem2}) and $\hat{\xi} \in
\mathcal{C}(x)$ be the optimizer to the primal problem
(\ref{primalProblem2}). Both $\hat{\eta}$ and $\hat{\xi}$ exist by Lemmas
\ref{ExistenceOfDualOptimizer} and \ref{existenceOfPrimalOptimizer}
respectively.  Using the definition of $ {V},$
(\ref{polarityCD}), (\ref{defCxDy}), and Theorem 23.5 in \cite{Rok} we get
\begin{equation}\nonumber
  0\leq  \int_{\Omega}\left(  {V}\left(\hat{\eta}\right) -  {U}\left(\hat{\xi}\right) +
    \hat{\xi}\hat{\eta}\right)d\mu
  \leq  v(y) - u(x) + xy=0.
\end{equation}
%where the latter equality holds by Theorem 23.5 in \cite{Rok}.
Therefore, for $\mu$ a.e. $\omega\in\Omega$  we have
\begin{equation}\nonumber
   {V}\left(\hat{\eta}\right) =  {U}\left(\hat{\xi}\right) -
  \hat{\xi}\hat{\eta}.
\end{equation}
This implies the remaining assertions of the theorem:
\begin{displaymath}
\begin{array}{c}
     {U}'\left(\hat{\xi}\right) = \hat{\eta}
     %\quad{\rm and}\quad\hat{\xi} = {I}\left(\hat{\eta}\right)
     \quad \mu~{\rm a.e.},\\
    \langle \hat{\xi},\hat{\eta}\rangle =
    \int_{\Omega}  {U}\left(\hat{\xi}\right)d\mu -
    \int_{\Omega}  {V}\left(\hat{\eta}\right)d\mu
    =u(x) - v(y)
    = xy.
\end{array}
\end{displaymath}
\end{proof}

%\subsection{The set $\mathcal{\tilde D}$ as the domain of the dual problem}
%%%%%%%%%%%%%%%%%%%%%%%%%%%%%%%%%%%%%%%%%%%%%%%%%%%%%%%%%%%%%%%%%%%%%%%%%%%%%
% \subsection{$v(y) = \inf\limits_{h\in \mathcal{\tilde D}
%   }\int_{\Omega}{V}(yh)d\mu$}
%%%%%%%%%%%%%%%%%%%%%%%%%%%%%%%%%%%%%%%%%%%%%%%%%%%%%%%%%%%%%%%%%%%%%%%%%%%%%%%%%%%%%%%
%%%%%%%%%%%%%%%%%%%%%%%Lemma 2 in KS2003 start%%%%%%%%%%%%%%%%%%%%%
In order to prove Theorem
\ref{secondTheorem2} we proceed in a way that is similar to the proof of Proposition 1
in Kramkov and Schachermayer \cite{KS2003}.
%The reader that is familiar with this proof might proceed
%directly to Section \ref{pfOfProp1}.
%We still present the proofs for the sake of completeness.
Define the \textit{polar} of a set $A\subseteq \mathbf{L}^0_{+}$ as
\begin{equation}\nonumber
  A^o\set\left\{ \xi\in \mathbf{L}^0_{+}:~ \langle \xi,\eta \rangle \leq 1~~{\rm
      for~all~} \eta\in  A \right\}.
\end{equation}
A subset $A$
of $\mathbf L^0_{+}$ is called \textit{solid} if $0\leq {\eta}\leq {\zeta}$ and
${\zeta}\in A$ implies that ${\eta}\in A$.
Observe that the sets $\mathcal C$ and $\mathcal D$ satisfy the bipolar relations.
We will use a version of the \textit{bipolar theorem} that was proven by Brannath
and Schachermayer in \cite{BranSchach}: for a subset $A$ of $\mathbf{L}^0_{+}$ the bipolar
$A^{oo}$ is the smallest subset of $\mathbf{L}^0_{+}$ containing $A,$
which is convex, solid, and closed with respect to the topology of
convergence in measure.

\begin{Lemma}\label{Lemma2}
  Under the conditions of Theorem \ref{mainTheorem2}, for every fixed
  $y>0$ let $\hat{{\eta}}(y)$ be the minimizer to the dual problem
  (\ref{dualProblem2}). Then there exists a sequence $\left(
    {\zeta}^n\right)_{n\geq 1}$ in $\tilde{\mathcal{D}}$ that  $\mu$ a.e. converges
  to $\hat{{\eta}}(y)/y$.
\end{Lemma}
\begin{proof} Fix $y>0.$ By assumption $\tilde{\mathcal{D}}$ is a convex set that satisfies
 (\ref{11191}). Therefore, applying the bipolar theorem (see \cite{BranSchach}) we deduce that $\mathcal{D}$ is the smallest
convex, closed and solid subset of $\mathbf{L}^0_{+}\left(\Omega,
\mathcal F, \mu\right)$ containing $\tilde{\mathcal{D}}$. Thus for
any ${\eta}\in\mathcal{D}$ there exists a sequence $\left({\zeta}^n\right)_{n\geq 1}$ in
$\tilde{\mathcal{D}}$ such that ${\zeta}=\lim\limits_{n\to\infty}{\zeta}^n$
exists $\mu$ a.e. and ${\zeta}\geq {\eta}$. In particular such a sequence
exists for ${\eta}=\hat{{\eta}}(y)/y$. We deduce from optimality of
$\hat{{\eta}}(y)$ that ${\eta}={\zeta}=\lim\limits_{n\to\infty}{\zeta}^n$.
%$\mu$ a.e.
\end{proof}

%%%%%%%%%%%%%%%%%%%%%%%%%%%%% Lemma 2 in KS2003 end

% The following two lemmas are adaptations of Proposition 1 in
% \cite{KS2003} to our settings.
% The reader familiar with the proof of
% this proposition might proceed directly to section
% \ref{subsection-conjugacy}. We present the proofs for the sake
% of completeness.
% % The proofs of the following two lemmas are
% %based on the techniques of Kramkov and Schachermayer \cite{KS2003}.
\begin{Lemma}\label{finitenessOverTildeY}
  Under the conditions of Theorem \ref{mainTheorem2} for each $y>0$
  we have
  \begin{equation}\nonumber
    \inf\limits_{{\eta}\in \mathcal{\tilde{D}}
    }\int_{\Omega}{V}(y{\eta})d\mu < \infty.
  \end{equation}

\end{Lemma}
\begin{proof} To simplify notations we will assume that $y=1$. Let
$\left({a}^n\right)_{n\geq 1}$ be a sequence of strictly positive
 numbers such that $\sum\limits_{n=1}^{\infty}{a}^n = 1$.  By
Lemma \ref{ExistenceOfDualOptimizer}, for each $n\geq 1$ there
exists $\hat{\eta}({{a}^n}),$ the minimizer to the dual problem
(\ref{dualProblem2}) when $y={a}^n$.
One can construct a sequence of strictly positive numbers
$\left(\delta_n\right)_{n\geq 2}$ that decreases to $0$, such that
\begin{equation}\label{aka29}
    \sum\limits_{n=1}^{\infty} \int_{\Omega}{V}\left(\hat{\eta}({a}^n)\right)1_{A_n}d\mu < \infty,
     {\rm ~~if} ~
      A_n\in \mathcal{F}, {\rm and~}\mu(A_n)\leq
      \delta_n,~n\geq 2.
\end{equation}
% \begin{equation}\label{aka29}
%   \begin{array}{c}
%     \sum\limits_{n=1}^{\infty} \int_{\Omega}{V}\left(\hat{\eta}({a}^n)\right)1_{A_n}d\mu < \infty,\\
%      {\rm if} ~
%       A_n\in \mathcal{F}, ~ {\rm and~}\mu(A_n)\leq
%       \delta_n,~n\geq 2.\\
%   \end{array}
% \end{equation}
From Lemma \ref{Lemma2} we deduce the existence of a sequence
$\left({\eta}^n\right)_{n\geq 1}\subset \tilde{\mathcal{D}}$ such that
\begin{equation}\nonumber%\label{tat}
  \mu\left({V}\left({a}^n{\eta}^n\right) > {V}\left(\hat{\eta}({a}^n)\right) +  1\right) \leq
  \delta_{n+1}, \hspace{2 mm} n\geq 1.
\end{equation}
Define the sequences of measurable sets $\left(B_n\right)_{n\geq 1}$
and $\left(A_n\right)_{n\geq 1}$ as follows:
\begin{equation}\nonumber%\label{tata}
  B_n \set \left\{{V}\left({a}^n {\eta}^n \right) \leq
    {V}\left(\hat{\eta}({a}^n) \right) + 1
  \right\}, \hspace{3mm} n\geq 1,
\end{equation}
\begin{equation}\nonumber%\label{tatat}
    A_1\set B_1,
     \dots ,
    A_n\set B_n
    \backslash \left( \bigcup\limits_{k=1}^{n-1} A_k\right),\dots.
 \end{equation}
Then $\left(A_n\right)_{n\geq 1}$ is a measurable partition of
$\Omega$ and $\mu\left(A_n\right)\leq \delta_n$ for $n\geq 2$.

To finish the proof, let ${\eta} \set
\sum\limits_{n=1}^{\infty}{a}^n{\eta}^n.$ Then
${\eta}\in\mathcal{\tilde{D}}$, since $\mathcal{\tilde{D}} $ is closed
under countable convex combinations.
% Using the monotone convergence
% theorem we get
% \begin{equation}\nonumber
%   \int_{\Omega} {V}({\eta})d\mu=
%   \sum\limits_{n=1}^{\infty}\int_{\Omega} {V}(
%   {\eta})1_{A_n} d\mu.
% \end{equation}
% Let ${V}^{+}$ and ${V}^{-}$ denote the positive and
% negative parts of ${V}.$ From monotone convergence theorem
% and (\ref{9141}) we get:
% \begin{equation}\nonumber
% \begin{array}{rcl}
% \sum\limits_{n=1}^{\infty}\int_{\Omega} {V}^{+}({\eta})1_{\{A_n\}}d\mu &=& \int_{\Omega} \sum\limits_{n=1}^{\infty} {V}^{+}({\eta})1_{\{A_n\}}d\mu \\
% &=& \int_{\Omega}{V}^{+}({\eta}) d\mu. \\
% \end{array}
% \end{equation}
% We can obtain a similar equality for ${V}^{-}.$ (\ref{621})
% follows.
From the construction of $\left(A_n\right)_{n\geq 1}$,
monotonicity of ${V}$, and (\ref{aka29}) we obtain
\begin{equation}\nonumber
  \begin{array}{rcl}
    \int_{\Omega}{V}({\eta})d\mu&=& \sum\limits_{n=1}^{\infty}\int_{\Omega}{V}\left(\sum\limits_{j=1}^{\infty}{a}^j{\eta}^j\right)1_{A_n}
    d\mu\\
    &\leq& \sum\limits_{n=1}^{\infty}\int_{\Omega}
    {V}\left({a}^n {\eta}^n\right)1_{A_n}
    d\mu\\
    &\leq& \sum\limits_{n=1}^{\infty}\int_{\Omega} {V}\left(\hat{\eta}({a}^n)\right)1_{A_n}
    d\mu + \mu(\Omega)\\
    &<& \infty.\\
  \end{array}
\end{equation}
This concludes the proof of the lemma.
\end{proof}

% %%%%%%%%%%%%%%%%%%%%%%%%%%%%%%%%%%%%%%%%%%%%%%%%%%%%%%%%%%%%%%%%%%%%%%%%%%%%%%%%%%%%%%%%%%
% % \subsection{"$v(y)= \inf\limits_{{\eta}\in \mathcal{\tilde{D}}
% %   }\mathbb{E} \left[ {V}(y{\eta})\right]$"}
% %%%%%%%%%%%%%%%%%%%%%%%%%%%%%%%%%%%%%%%%%%%%%%%%%%%%%%%%%%%%%%%%%%%%%%%%%%%%%%%%%%%%%%%%%%
\begin{proof}[Proof of Theorem \ref{secondTheorem2}] By symmetry between
  the primal and dual problems, it suffices to prove that
  \begin{equation}\nonumber%\label{63-1}
    v(y)= \inf\limits_{{\eta}\in \mathcal{\tilde{D}} }\int_{\Omega}
    {V}(y{\eta})d\mu,\quad y>0.
  \end{equation}
Fix $y >0$ and $\epsilon > 0$. We will show that there exists ${\eta}
\in \mathcal{\tilde{D}} $ such that
\begin{equation}\nonumber%\label{630}
  \int_{\Omega}{V}\left((y+\epsilon ){\eta}\right)d\mu
  \leq
  v(y) + \epsilon.
\end{equation}
Let $\hat{\eta}\in\mathcal{D}(y)$ be the minimizer to the dual
problem (\ref{dualProblem2}), ${\zeta}$ be an element of
$\mathcal{\tilde{D}} $, such that $$\int_{\Omega}
{V}\left({\epsilon \zeta}\right)d\mu < \infty,$$ whose existence
follows from Lemma \ref{finitenessOverTildeY}. Let $\delta > 0$ be  such that
\begin{displaymath}
    \int_{\Omega}
    \left(
      \left|{V}\left(\hat{\eta}\right)\right|+
      \left|{V}\left({\epsilon \zeta}\right)\right|
    \right)1_{A}
    d\mu\leq \frac{\epsilon}{2},
      {\rm ~~if}~
    A\in \mathcal{F} ~{\rm with}~
    \mu(A)\leq \delta.
\end{displaymath}
% \begin{equation}\nonumber%\label{631}
%   \begin{array}{c}
%     \int_{\Omega}
%     \left(
%       \left|{V}\left(\hat{\eta}\right)\right|+
%       \left|{V}\left({\epsilon \zeta}\right)\right|
%     \right)1_{A}
%     d\mu\leq \frac{\epsilon}{2},\\
%     {\rm if}~
%     A\in \mathcal{F} ~{\rm with}~
%     \mu(A)\leq \delta.
%   \end{array}
% \end{equation}
By Lemma \ref{Lemma2} there exists $\theta \in
\tilde{\mathcal{D}}$ such that the set
\begin{equation}\nonumber%\label{632}
  B\set \left\{
    {V}\left({y \theta}\right) >
    {V}\left(\hat{\eta}\right) + \frac{\epsilon}{2\mu(\Omega)}
  \right\}
\end{equation}
has measure $\mu(B)\leq \delta$. Define
\begin{displaymath}
 {\eta}\set \frac{{y \theta} + \epsilon
\zeta}{y + \epsilon}.
\end{displaymath}
Since $\tilde{\mathcal{D}}$ is convex it follows that ${\eta}\in
\tilde{\mathcal{D}}.$ By construction of the set $B$ and monotonicity of ${V}$ we obtain
\begin{equation}\nonumber
  \begin{array}{rcl}
    \int_{\Omega}
    {V}\left((y+ \epsilon) {\eta}\right)
    d\mu &=& \int_{\Omega}
    {V}\left({y \theta} + {\epsilon \zeta}\right)d\mu\\
    % &=& \int_{\Omega}{V}\left({y \theta} + {\epsilon \zeta}\right)\left( 1_A +
%       1_{A^c} \right) d\mu\\
    &\leq&  \int_{\Omega} {V}\left({y \theta}\right)1_{B^c}d\mu +
    \int_{\Omega}{V}\left({\epsilon \zeta}\right)1_B
    d\mu \\
    &\leq&\frac{\epsilon}{2} +
    \int_{\Omega}{V}\left(\hat{\eta}\right)d\mu+  \int_{\Omega}\left({V}\left({\epsilon \zeta}\right)-{V}\left(\hat{\eta}\right)\right)1_{B}
    d\mu\\
    &\leq& v(y) + \epsilon.
  \end{array}
\end{equation}
% i.e., (\ref{630}) holds.  Let $z \set y+ \epsilon$. Then (\ref{630})
% implies:
% \begin{equation}\nonumber
%   \int_{\Omega}{V}\left(z {\eta}\right)d\mu \leq v(z-\epsilon ) + \epsilon.
% \end{equation}
% Taking the limit as $\epsilon \downarrow 0$ and using continuity of
% $v$ (by convexity) we deduce~(\ref{63-1}).
\end{proof}

\section{Proofs of the main theorems}\label{pfOfProp1}

\looseness-1 Let us recall the concept of Fatou convergence of stochastic
processes, see \cite{FK}.
\begin{Definition} Let $\tau$ be a dense subset
  of $[0,\infty)$. A sequence of processes $(Y^n)_{n\geq 1}$ is \textit{Fatou
    convergent on $\tau$} to a process $Y$, if   $(Y^n)_{n\geq 1}$ is
  uniformly bounded from below and
\begin{equation}\nonumber
Y_t = \limsup\limits_{s\downarrow
  t,~s\in\tau}\limsup\limits_{n\to\infty}Y^n_s = \liminf\limits_{s\downarrow
  t,~s\in\tau}\liminf\limits_{n\to\infty}Y^n_s
\end{equation}
almost surely for every $t\geq 0$. If $\tau = [0,\infty)$, then the sequence
$(Y^n)_{n\geq 1}$ is called \textit{Fatou convergent}.
\end{Definition}
We also recall that a probability measure $\mathbb Q$ is called an \textit{equivalent
  local martingale measure} for $\mathcal X$, if $\mathbb Q$ is equivalent to $\mathbb P$ and
every $X\in\mathcal X$ is a local martingale under $\mathbb Q$. We denote the
set of equivalent local martingale measures by $\mathcal M^e$.

 The following lemma can be thought as an extension of Theorem 5.12 in
 \cite{Delbaen-Schachermayer1998} to our settings. The proof of Lemma \ref{characterizationOfAdmissibleConsumptionProcess} is
 based on an application of Fatou convergence and the
 optional decomposition theorem, see
\cite{K96, FK}. However, since assumption (\ref{ZisNotEmpty}) is weaker than
the condition $\mathcal M^e\neq \emptyset$ in \cite{K96, FK}, we need to
do extra work.

\begin{Lemma}\label{characterizationOfAdmissibleConsumptionProcess}
Let $c$ be a nonnegative optional process and $\kappa$ be a stochastic clock.
Under the assumptions (\ref{stochasticClock}) and (\ref{ZisNotEmpty}), the following conditions are equivalent:
\begin{enumerate}[(i)]
\item
$
c\in\mathcal A,
$
\item
$
\sup\limits_{Z\in{\mathcal Z}}\mathbb
E\left[\int_0^{\infty}c_tZ_td\kappa_t
\right]\leq 1.
$
\end{enumerate}
\end{Lemma}
\begin{proof}
Let $c\in\mathcal A$. Then there exists a predictable
$S$-integrable process $H$, s.t.
\begin{displaymath}
1+\int_0^t H_udS_u \geq \int_0^tc_ud\kappa_u\geq 0,\quad t\geq 0.
\end{displaymath}
Take an arbitrary $Z\in\mathcal Z$. Using supermartingale property of
$Z_t(1+\int_0^t H_udS_u)$, $t\geq 0$, we obtain for every $T\geq 0$
\begin{displaymath}
1\geq \mathbb E\left[Z_T\left(1+\int_0^T H_udS_u\right)\right] \geq \mathbb E\left[Z_T\int_0^Tc_ud\kappa_u\right].
\end{displaymath}
Using localization and integration by parts we deduce
\begin{displaymath}
\mathbb E\left[Z_T\int_0^Tc_ud\kappa_u\right] = \mathbb
E\left[\int_0^Tc_uZ_ud\kappa_u
\right].
\end{displaymath}
Taking $T\to\infty$ and using the monotone convergence theorem, we get $(ii)$.

Conversely, assume that $\sup\limits_{Z\in{\mathcal Z}}\mathbb
E\left[\int_0^{\infty}c_tZ_td\kappa_t
\right]\leq 1$.
Using localization and integration by parts we deduce from $(ii)$:
\begin{displaymath}
\mathbb E\left[Z_n\int_0^nc_ud\kappa_u\right] = \mathbb
E\left[\int_0^nc_uZ_ud\kappa_u
\right],\quad n\geq 0.
\end{displaymath}

One can see that $\left\{(Z_t)_{t\in[0,n]}:~Z\in\mathcal Z \right\}$
coincides with the set of c\'adl\'ag densities of equivalent local martingale
measures for $\mathcal X$ on
$\left(\Omega, \mathcal F_n\right)$. Let us denote the set of such measures by
$\mathcal M^e_n$.
Then, by Proposition 4.2 in \cite{K96}, there exists a c\'adl\'ag
process $V^n$ on $[0,n]$ given by
\begin{displaymath}
V^n_t=\esssup\limits_{\mathbb Q\in\mathcal M^e_n}\mathbb E^{\mathbb Q}\left[\int_0^nc_ud\kappa_u|\mathcal
F_t\right],\quad t\in[0,n],
\end{displaymath}
which is a supermartingale under every $\mathbb Q\in\mathcal M^e_n$.
Notice that $V^n_t\geq \int_0^tc_ud\kappa_u$, $t\in[0,n]$, and $V^n_0\leq 1$. Now, applying Theorem 4.1 in \cite{FK}, we
can write $V^n$ as
\begin{displaymath}
V^n_t= V^n_0 + \int_0^tH^n_udS_u - A^n_t, \quad t\in[0,n],
\end{displaymath}
where $H^n$ is predictable $S$-integrable and $A^n$ is optional and
increasing, s.t. $A^n_0=0$.  Let us extend $H^n$ to $[0,\infty)$ by setting
$H^n_t \set 0  $ for
$t>n$. Using Lemma 5.2 in \cite{FK}, we can construct a sequence of stochastic
processes $Y^n\in {\rm conv}\left(1+\int_0^{\cdot}H^n_udS_u,
1+\int_0^{\cdot}H^{n+1}_udS_u,\dots\right)$,~$n\geq 1$, and a process $Y$, such that
$(ZY^n)_{n\geq 1}$ is Fatou
    convergent on the set of positive rational numbers to
a supermartingale $ZY$ for every $Z\in\mathcal Z$. Then, we have $Y_t \geq
\int_0^tc_ud\kappa_u$, $t\geq 0$, and $Y_0\leq 1$. Now, on $[0,n]$ using
Theorem 4.1 in \cite{FK}, we get
\begin{displaymath}
Y_t = Y_0 + \int_0^tG^n_udS_u - B^n_t,  \quad t\in[0,n],
\end{displaymath}
where $G^n$ is predictable $S$-integrable and $B^n$ is optional and increasing with
$B^n_0=0$.
Let us set $G^n_t \set 0$ for $t>n$. Denoting
\begin{displaymath}
n(t) \set \min\left\{n\in\mathbb N:~n>t \right\}, \quad t\geq 0,
\end{displaymath}
we deduce that the process
\begin{displaymath}
\tilde G_t\set \sum\limits_{k=1}^{n(t)}\left( G^k_t -
  G^{k-1}_t\right),\quad t\geq 0,
\end{displaymath}
  is such that $1+\int_0^t\tilde G_udS_u\geq
\int_0^tc_ud\kappa_u$, $t\geq 0$. Thus, $c\in\mathcal A$.
\end{proof}

\begin{Lemma}\label{4102}
  Let $\kappa$ be a stochastic clock.
  Under the assumptions
  (\ref{stochasticClock}) and (\ref{ZisNotEmpty}),
% ${\mathcal Z}$ is a
%   subset of ${\mathcal Y}$ that is
%   closed under countable convex combinations and is such that
for every
  $c\in\mathcal A$ we have
  \begin{equation}\nonumber%\label{4101}
    \sup\limits_{Z\in{\mathcal Z} }\mathbb{E}\left[ \int_0^{\infty} c_tZ_td\kappa_t \right] =
    \sup\limits_{Y\in\mathcal{ Y} }\mathbb{E}\left[ \int_0^{\infty} c_tY_td\kappa_t
    \right]\leq 1.
  \end{equation}
\end{Lemma}
\begin{proof}
 By definition (\ref{defOfYkappa}) for an arbitrary
$Y\in{\mathcal Y}$ we can find a sequence $\left(Y^n\right)_{n\geq
  1}$ in the solid hull of ${\mathcal Z}$ (i.e., such that $Y^n\leq Z^n$
$\left( d\kappa\times\mathbb P\right)$ a.e. for some $Z^n\in{\mathcal Z}$), such
that $\left(Y^n\right)_{n\geq
  1}$ converges $\left(d\kappa\times\mathbb P\right)$ a.e. to $Y$.
 Using Fatou's lemma and Lemma
\ref{characterizationOfAdmissibleConsumptionProcess} we get
\begin{equation}\nonumber
  \mathbb{E}\left[ \int_0^{\infty} c_tY_td\kappa_t
  \right]
  \leq \liminf\limits_{n\to\infty}\mathbb{E}\left[ \int_0^{\infty}
    c_tY^n_td\kappa_t \right] \leq
  \sup\limits_{Z\in{\mathcal Z} }\mathbb{E}\left[ \int_0^{\infty} c_tZ_td\kappa_t
  \right]
  \leq 1.
\end{equation}
\end{proof}

Denote by $\mathbf{L}^0=\mathbf{L}^0\left(
  d\kappa\times\mathbb{P}\right)$ the linear space of (equivalence classes of)
real-valued optional processes on
the stochastic basis $\left(\Omega, \mathcal{F},
  \left(\mathcal{F}_t\right)_{t\geq 0}, \mathbb{P}\right)$ which we equip
with the topology of
convergence in measure $\left(d\kappa\times \mathbb P\right)$. Let
$\mathbf{L}^0_{+}$ be the positive orthant of $\mathbf{L}^0.$ Recall that a
\textit{polar} of a set $A\subseteq \mathbf{L}^0_{+}$ is defined as:
\begin{equation}\nonumber
  A^o\set\left\{ Y\in \mathbf{L}^0_{+}:\hspace{1mm} \mathbb{E}\left[\int_0^{\infty} c_tY_td\kappa_t \right] \leq 1~{\rm
      for\hspace{1mm} all} \hspace{1mm} c\in  A \right\}.
\end{equation}
In view of Theorems \ref{mainTheorem2} and \ref{secondTheorem2} in order to
complete the proofs of Theorems \ref{mainTheorem} and \ref{secondTheorem} it
suffices to establish the following proposition. Note that the sets
$\mathcal C$, $\mathcal D$ and measure $\mu$
correspond to the sets $\mathcal A$, ${\mathcal Y}$ and measure $\left(
  d\kappa\times\mathbb P\right)$,
 the sets $\tilde{\mathcal C}$ and $\tilde{\mathcal
   D}$ accord with the sets ${\mathcal B}$ and $\mathcal Z$, respectively.

\begin{Proposition}\label{prop1}
  Assume that an $\mathbb{R}^d$-valued semimartingale $S$ satisfies
  (\ref{ZisNotEmpty}).
 %Let $\kappa$ be a stochastic clock.
  Under the condition (\ref{stochasticClock}), the sets $\mathcal{A}$ and
  $\mathcal{Y},$ defined in (\ref{defOfAcal}) and
  (\ref{defOfYkappa}), respectively, have the following properties:

  (i) $\mathcal{A}$ and $\mathcal{Y}$ are subsets of
  $\mathbf{L}^0_{+}$ that are convex, solid and closed in the topology
  of convergence in measure $\left(d\kappa\times\mathbb P\right).$

  (ii) The sets $\mathcal{A}$ and $\mathcal{Y}$
  satisfy the bipolar relations:
  \begin{equation}\nonumber%\label{bipolarIneq}
    \begin{array}{lclcl}
      c\in\mathcal{A}&{\Leftrightarrow}& \mathbb{E}\left[ \int_0^{\infty}c_tY_td\kappa_t \right]\leq 1& {\rm for
        ~all} & Y\in\mathcal{Y}, \\
      Y\in\mathcal{Y}&{\Leftrightarrow}& \mathbb{E}\left[\int_0^{\infty} c_tY_t d\kappa_t\right]\leq 1& {\rm for
        ~all}& c\in\mathcal{A}.\\
    \end{array}
  \end{equation}

  (iii)% $\mathcal{A}$ is a bounded subset of $\mathbf{L}^0$.
There exists $c\in\mathcal{A}$ such that $c>0$ and there
  exists $Y\in\mathcal{Y}$ such that $Y>0.$
\end{Proposition}

%%%%%%%%%%%%%%%%%%%%%%%%%%%%%%%%%
%%%%%%%%%%%%%%%%%%%%%%%%%%%%%%%%%%
 %%%%%%%%%%%%%%%%%%%%%%%%%%%%%%%%%%
\begin{proof} \textit{(i)} It is enough to show closedness of
$\mathcal{A}.$ Let $\left(c^n\right)_{n\geq 1}$ be a
sequence in $\mathcal{A}$ that $\left(d\kappa\times\mathbb P\right)$ a.e. converges to $~c$.
  For an arbitrary
$Z\in\mathcal{Z}$ using Fatou's lemma and
Lemma \ref{characterizationOfAdmissibleConsumptionProcess} we get:
\begin{equation}\nonumber
  \mathbb{E}\left[\int_0^{\infty}c_tZ_td\kappa_t\right]
   \leq
  \liminf\limits_{n\to\infty}\mathbb{E}\left[\int_0^{\infty}c^n_tZ_td\kappa_t\right]\leq 1.
\end{equation}
 Therefore by Lemma \ref{characterizationOfAdmissibleConsumptionProcess}, $c\in\mathcal
 A$, and thus $\mathcal{A}$ is closed.

\textit{(ii)} It follows from Lemma
\ref{characterizationOfAdmissibleConsumptionProcess} that
% $\mathcal{A}\subseteq\mathcal{Z}^o.$ On the other
% hand for any $g\in\mathcal{Z}^o$ we have
% $\mathbb{E}\left[\int_0^{\infty} g_tZ_t d\kappa_t\right]\leq 1$ for
% all $Z\in{\mathcal Z}$. By Lemma
% \ref{characterizationOfAdmissibleConsumptionProcess}, $g\in\mathcal
% A$, consequently $\mathcal{Z}^o\subseteq\mathcal{A}.$ As a
% result,
\begin{equation}\nonumber
  \mathcal{A} = \mathcal{Z}^o,
\end{equation}
whereas from Lemma \ref{4102} we deduce
\begin{equation}\label{11192}
  \mathcal{Y} \subseteq \mathcal{A}^o
  =\mathcal{Z}^{oo}.
\end{equation}
Since $\mathcal{Y}$ is closed, convex, and solid and
$\mathcal{Z}\subset \mathcal{Y},$ it follows from the
bipolar theorem of Brannath and Schachermayer that $\mathcal{Z}^{oo}\subseteq
\mathcal{Y}.$ Combining this with (\ref{11192}) we conclude
that
% $\mathcal{Y} \subseteq
% \mathcal{A}^o =\mathcal{Z}^{oo}\subseteq
% \mathcal{Y}.$ Consequently,
\begin{equation}\label{11194}
  \mathcal{Y} = \mathcal{A}^o.
\end{equation}
On the other hand it follows from part \textit{(i)} that $\mathcal{A}$
is also convex, closed and solid. Thus $\mathcal{A} =
\mathcal{A}^{oo}$ by the bipolar theorem. Therefore, from (\ref{11194})
we get
\begin{displaymath}
 \mathcal{A} =\mathcal{Y}^o.
\end{displaymath}

\textit{(iii)}
Since $\mathcal{X}$ contains a constant function $\mathbf{1} =
\left(1\right)_{t\geq 0}$, the existence of $c\in \mathcal{A},$ such that $c>0,$
follows from the definition of the set $\mathcal{A}$. The existence of $Y\in
\mathcal{Y},$ such that $Y>0,$ follows from assumption
(\ref{ZisNotEmpty}).  This completes the proof of
Proposition~\ref{prop1}.
\end{proof}

\section*{Acknowledgements}
%\textbf{Acknowledgements.}
This work is  part of the author's PhD Thesis. I
would like to thank Dmitry Kramkov for being a supportive adivisor, whose
suggestions and ideas helped in writing this paper. I would also like to thank
 Giovanni Leoni, Scott Robertson, and Pietro Siorpaes for the remarks and discussions.

%\addcontentsline{toc}{section}{Bibliography}

\bibliographystyle{plainnat} \bibliography{finance}

\end{document}